\documentclass[12pt]{article}

\usepackage[pdftex]{graphicx}
\usepackage{amssymb}
\usepackage{amsfonts}
\usepackage{graphicx}
\usepackage[margin=2.7cm]{geometry}
\usepackage{amsmath}
\usepackage{euscript}
\usepackage{enumerate}
\usepackage{amsthm}
\usepackage{setspace} 
\usepackage{hyperref}

\newtheorem{theorem}{Theorem}[section]
\newtheorem{definition}[theorem]{Definition}
\newtheorem{lemma}[theorem]{Lemma}

\newtheorem{proposition}[theorem]{Proposition}
\newtheorem{example}[theorem]{Example}

\newtheorem{remark}[theorem]{Remark}

\numberwithin{equation}{section}
\def\unnumfootnote{\xdef\@thefnmark{}\@footnotetext}

\newcommand{\Rmnum}

\title{Hyperfunctions and Spectral Zeta Functions of Laplacians on Self-Similar Fractals}

\author{Nishu Lal, Michel L. Lapidus \\
\small Department of Mathematics, University of California, Riverside, CA 92521-0135, USA\\
\small email: nishul@math.ucr.edu, lapidus@math.ucr.edu\\
}

\begin{document}
\maketitle

\begin{abstract}
\noindent We investigate the spectral zeta function of fractal differential operators such as the Laplacian on the unbounded (i.e., infinite) Sierpinski gasket and a self-similar Sturm--Liouville operator associated with a fractal self-similar measure on the half-line.  In the latter case, C. Sabot discovered the relation between the spectrum of this operator and the iteration of a rational map of several complex variables, called the renormalization map.  We obtain a factorization of the spectral zeta function of such an operator, expressed in terms of the Dirac delta hyperfunction, a geometric zeta function, and the zeta function associated with the dynamics of the corresponding renormalization map, viewed either as a polynomial function on $\mathbb{C}$ (in the first case) or (in the second case) as a polynomial on the complex projective plane, $\mathbb{P}^2(\mathbb{C})$.  Our first main result extends to the case of the fractal Laplacian on the unbounded Sierpinski gasket a factorization formula obtained by the second author for the spectral zeta function of a fractal string and later extended by A. Teplyaev to the bounded (i.e., finite) Sierpinski gasket and some other decimable fractals.  Furthermore, our second main result generalizes these factorization formulas to the renormalization maps of several complex variables associated with fractal Sturm--Liouville operators.  Moreover, as a corollary, in the very special case when the underlying self-similar measure is Lebesgue measure on $[0, 1]$, we obtain a representation of the Riemann zeta function in terms of the dynamics of a certain polynomial in $\mathbb{P}^2(\mathbb{C})$, thereby extending to several variables an analogous result by A. Teplyaev.  
\vspace{5mm}

\let\thefootnote\relax

\end{abstract}

\newpage
\section{Introduction}

Connections between analysis on fractals, spectral theory, and complex dynamics have been of great interest in physics and mathematics in the past few decades.  See, for example, \cite{AO82}--\cite{BC08}, \cite{B79}, \cite{HHW87}, \cite{NY03}--\cite{RT83} and \cite{B91}, \cite{DGV08}, \cite{UF03}--\cite{FS92}, \cite{Kigami93}--\cite{Lapidus93}, \cite{LvanF06}, \cite{CS98}--\cite{CS03}, \cite{TS96}, \cite{RS06}, \cite{AT04}, \cite{AT07}, as well as the references therein.  In this paper, we further investigate these connections by focusing on the model studied by C. Sabot (\cite{CS98}-\cite{CS05}) in the context of multivariable complex dynamics.  It involves singular diffusions as well as fractal Hamiltonians and their spectra, along with the associated complex dynamics.

We show that the spectral zeta function of the fractal Laplacian on an infinite (unbounded) Sierpinski gasket  and the spectral zeta function of certain fractal Sturm--Liouville differential operators have a special factorization involving the Dirac hyperfunction, a geometric zeta function, and a zeta function associated with the renormalization map induced by the decimation method.  The decimation method is a process that describes the interesting relations between the spectrum of a differential operator on a suitable self-similar fractal and the dynamics of the iteration of some complex polynomial (or rational function).  First we focus on the infinite Sierpinski gasket (a suitable deterministic blow-up of the bounded gasket) and show that the spectral zeta function of the Laplacian can be written as the product of the Dirac delta hyperfunction and the spectral zeta function associated with the bounded Sierpinski gasket, which is further expressed by A. Teplyaev \cite{AT07} as the product of the zeta function of a quadratic polynomial in one complex variable and a suitable geometric zeta function (in the sense of \cite{LvanF06}).  Second, we study the Sturm--Liouville operator associated with a fractal self-similar measure on the half-line and its spectral properties associated with the renormalization map on the complex projective plane, $\mathbb{P}^2(\mathbb{C})$.  We define the zeta function associated with the renormalization map of several complex variables and observe that in a very special case corresponding to Lebesgue measure on $[0, 1]$, the Riemann zeta function can be represented in terms of it.

The theory of fractal strings and the corresponding geometric and spectral zeta functions were first studied by the second author and his collaborators in \cite{Lapidus91}--\cite{LP93}; see also \cite{LvanF06}  for a detailed exposition.  A fractal string, $\mathcal{L}$, is a countable collection of disjoint intervals of lengths $\ell_j$.  The spectral zeta function of the Dirichlet Laplacian $L$ on $\mathcal{L}$ has the following factorization 

\begin{equation}
\label{La1}
\zeta_{L}(s)= \pi ^{-s}\zeta(s)\zeta_{\mathcal{L}}(s),
\end{equation}
where $\zeta_{\mathcal{L}}(s)=\sum_{j=1}^{\infty} \ell_j^s$ is the geometric zeta function of the fractal string $\mathcal{L}$ and $\zeta(s)=\sum_{n=1}^{\infty} n^{-s}$ is the Riemann zeta function (or its meromorphic continuation).  See \cite{Lapidus92}, \cite{Lapidus93}, \cite{LM95}, \cite{LP93} and \cite{LvanF06}, Thm. 1.19.

Later on, A. Teplyaev (\cite{AT04}, \cite{AT07}) studied the spectral zeta function of the Laplacian on a class of symmetric finitely ramified fractals which includes the classic Sierpinski gasket.  He discovered that the factorization of the spectral zeta function involves a new zeta function associated with a polynomial (see  \cite{AT07}), thereby naturally extending the factorization for fractal strings.  For $Re(s) > d_R=\frac{2 \log N}{\log c}$, he defined the zeta function of a complex polynomial of degree $N$, $R=R(z)$, to be 
\[
 \zeta_{R,z_0} (s) = \lim_{n\rightarrow \infty}\sum_{z\in R^{-n}\{z_0\}}(c^n z)^{-\frac{s}{2}},
\]
assuming that $\mathcal{J}$, the Julia set of $R$, satisfies $\mathcal{J}\subset [0,\infty)$, $R(0)=0$ and $c=R'(0) > 1$.  He then showed that this zeta function has a meromorphic continuation to the half-plane. Within the one variable case, our first main result shows that the spectral zeta function of the Laplacian of the infinite Sierpinski gasket has a factorization formula expressed in terms of the spectral zeta function of the finite Sierpinski gasket and the Dirac delta hyperfunction; see Thm. \ref{infiniteSGzeta}, along with Thm \ref{SG}.  

Providing a generalization of the decimation method for a certain class of fractals, Sabot (\cite{CS98}--\cite{CS03}) introduced a rational function of several complex variables, $\rho$, called the renormalization map.  In particular, following \cite{CS98}--\cite{CS05}, we consider the Sturm--Liouville operator $H=-\frac{d}{dm}\frac{d}{dx}$ on the interval $I=[0,1]$, where $m$ is induced by a self-similar measure.  The Sturm--Liouville operator $H_{<\infty>}$ on $[0,\infty)$ is viewed as a limit of the sequence of operators $H_{<n>}=-\frac{d}{dm_{<n>}}\frac{d}{dx}$ with Dirichlet boundary conditions on $I_{<n>}=[0,\alpha^{-n}]$ which are the infinitesimal generators of the Dirichlet forms $(a_{<n>}, m_{<n>})$.  (Here, $\alpha \in (0, 1)$ is a suitable parameter; see \S 3.1.  Furthermore, note that for $n=0$, $H_{<0>}=H$ and $I_{<0>}=I$.)  The invariant curve $\phi$ of the map $\rho$ is defined in terms of the trace of the Dirichlet form on a finite set.  The dynamics of $\rho$ on this invariant curve is used to compute the spectrum of the corresponding operator.  (See \S3.1 for a review of the general framework and for a discussion of the associated decimation method.)

In \S3.2, we first define (for Re(s) sufficiently large) the zeta function $\zeta_{\rho}$ associated with the renormalization map on the complex projective plane $\mathbb{P}^2(\mathbb{C})$ as 

\[
\zeta_{\rho}(s) = \sum_{p=0}^{\infty}\sum_{\{\lambda \in \mathbb{C}: \mbox{ }\rho^p(\phi(\gamma^{-(p+1)}\lambda)) \in D\}} (\gamma^p \lambda)^{\frac{-s}{2}}, 
\]
where $D$ is a suitable subset of the Fatou set of $\rho$, and relate it to the spectral zeta function $\zeta_{sp}$ of the Sturm--Liouville operator, via a product formula of the form

\begin{equation}
\label{La2}
\zeta_{sp}(s)=\zeta_{\rho}(s) \zeta_{\mathcal{L}}(s).
\end{equation}
Here, $\zeta_{\mathcal{L}}$ is the geometric zeta function of some underlying fractal string $\mathcal{L}=\{\ell_j\}_{j=1}^{\infty}$, viewed as a sequence of scales naturally associated with the Sturm--Liouville problem.

Each element of the sequence of spectral zeta funtions $\zeta_{H_{<n>}}$ of $H_{<n>}$ on $I_{<n>}=[0,\alpha^{-n}]$ has a product formula, and the most interesting cases are $\zeta_{H_{<0>}}$ and $\zeta_{H_{<\infty>}}$ on $I$ and $\mathbb{R}^+ = [0, \infty)$, respectively.   Indeed, we will show that the zeta function $\zeta_{H_{<0>}}(s)$  is exactly equal to $\zeta_{\rho}(s)$,
\[
\zeta_{H_{<0>}}(s):= \sum_{n=1}^{\infty} \sum_{p=0}^{\infty} (\gamma^p \lambda_n)^{-\frac{s}{2}} =\zeta_{\rho}(s);
\]
see Thm. \ref{H0MultiPolyZeta} and Prop. \ref{HnZetaFunction}.  We will also show that due to the spectral behavior of $H_{<\infty>}$ on $\mathbb{R}^+$, we can factor out the spectral zeta function $\zeta_{H_{<\infty>}}$ as the product of the Dirac hyperfunction (see \S 2.1) and of the zeta function produced by the generating set $S=\{\lambda_n\}^{\infty}_{n=1}$ of the spectrum associated with $\rho$; see Thm. \ref{HZetaFunction}.

The ordinary Dirichlet Laplacian $-\frac{d^2}{dx^2}$ is a special case of the Sturm--Liouville operator when $\alpha =\frac{1}{2}$ and hence, $m$ is the Lebesgue measure.   In that case, we will show in Thm. \ref{RiemannZetaFunction} that the Riemann zeta  function $\zeta$ can be expressed in terms of the zeta function, $\zeta_{\rho}$, associated with the renormalization map $\rho$.  This result extends to several complex variables the corresponding result by A. Teplyaev stating that the Riemann zeta function can be written in terms of the zeta function of a polynomial of one complex variable. 

We expect that some of the results obtained in this paper and the techniques developed in the process should be relevant to the study of diffusion on fractals \cite{B91} and aspects of condensed matters physics (see, e.g.,  \cite{AO82}, \cite{B79}, \cite{HHW87}, \cite{NY03}--\cite{RT83}).

\section{The spectral zeta function and the zeta function of a polynomial of one variable}
\numberwithin{equation}{section}

We devote this section to the Sierpinski gasket $SG$, a classical example of a self-similar fractal on which the Laplacian is widely explored.  The Sierpinski gasket is the unique nonempty compact subset of $\mathbb{R}^2$ such that $SG=\Psi_1(SG) \cup \Psi_2(SG) \cup \Psi_3(SG)$, where $\Psi_1(x, y)=\frac{1}{2}(x, y)$, $\Psi_2(x, y)=\frac{1}{2}(x, y) + (\frac{1}{2}, 0)$ and $\Psi_3(x, y)=\frac{1}{2}(x, y)+(\frac{1}{4}, \frac{\sqrt{3}}{4})$ are contraction mappings on $\mathbb{R}^2$.  Kigami (\cite{Kigami93}, \cite{Kigami01}) has defined the Laplacian on $SG$ (or on more general p.c.f. self-similar fractals) as a limit of Laplacians on a sequence of approximating finite graphs $\Gamma_m$ with a set of vertices $V_m$.  Given a function $u$ on $\Gamma_m$ (resp., SG), define the renormalized graph energy
\[
\mathcal{E}_m(u,u) =\left(\frac{5}{3}\right)^m \sum_{x_{\sim} y} \left(u(x)-u(y)\right)^2
\]
and then, for $u$ defined on SG, let $\mathcal{E}(u,u) = \lim_{m\rightarrow\infty} \mathcal{E}_m(u,u)$.  (Here, the sum is extended to all unordered pairs $\{x, y\}$ of neighboring vertices on $\Gamma_m$.)  Given the natural self-similar measure $\mu$ on $SG$, the equation $\Delta_{\mu}u = f$ can be interpreted variationally as $\mathcal{E}(u,v)$ = $-\int_{SG} fv d\mu$, for all suitable functions $v$ vanishing on the boundary points.  The pointwise formula for the Laplacian $\Delta_{\mu}$ is given by
\[
\Delta_{\mu}u(x) = \frac{3}{2}\lim_{m \rightarrow \infty}5^m \Delta_mu(x).
\]

The physicists R. Rammal \cite{RR84} and R. Rammal and G. Toulouse \cite{RT83} have studied the spectrum of the Laplacian, in particular, the eigenvalue equation $\Delta_{\mu}u = \lambda u$, and discovered the decimation method which establishes the relations between the spectrum of the Laplacian and the dynamics of the iteration of some polynomial $R$ on $\mathbb{C}$.  Later on, T. Shima \cite{TS96} and M. Fukushima and T. Shima \cite{FS92} gave a precise mathematical statement of their results, as follows (see also the exposition in \cite{RS06}):

\begin{theorem} [Fukushima--Shima, \text{\cite{FS92}, \cite{TS96}}]
\label{FuSh}
If $u$ is an eigenfunction of $-\Delta_{m+1}$ with eigenvalue $\lambda$, that is, $-\Delta_{m+1}u = \lambda u$, and if $\lambda \notin B$, then $-\Delta_m (u|_{V_m}) = R(\lambda)u|_{V_m}$, where $B = \{\frac{5}{4}, \frac{1}{2}, \frac{3}{2}\}$ is the set of `forbidden' eigenvalues and $R(z)=z(5-4z)$.  Conversely, if $-\Delta_m u = R(\lambda) u$ and $\lambda\notin B$, then there exists a unique extension $\tilde{u}$ of $u$ such that
$-\Delta_{m+1}\tilde{u}= \lambda\tilde{u}$.

\end{theorem}

In \cite{AT07}, Teplyaev studied the spectral properties of the Laplacian on SG and explored interesting connections between the spectral zeta function and the iteration of the polynomial $R$ (of a single complex variable) induced by the decimation method.  (See also \cite{AT04}.)  The \textit{spectral zeta function} is the meromorphic continuation of the Dirichlet series defined in terms of the eigenvalues of the Laplace operator, as we now explain.

\begin{definition}
\label{SpecZetaFunction}
The \emph{spectral zeta function} of a positive self-adjoint operator $L$ with compact resolvent \emph{(}and hence, with discrete spectrum\emph{)} is given \emph{(}for $Re(s)$ large enough\emph{)} by 
\begin{equation}
\label{SpecZeta}
\zeta_L(s) = \sum^{\infty}_{j=1}(\kappa_j)^{-s/2},
\end{equation}
where the positive real numbers $\kappa_j$ are the eigenvalues of the operator $L$ written in nondecreasing order and counted according to their multiplicities. 
\end{definition}


\begin{definition}[Teplyaev, \text{\cite{AT07}}]
\label{PolyZeta}
Let $R$ be a polynomial of degree $N$ satisfying $R(0)=0$, $c:=R^{\prime}(0) > 1$, and with Julia set $\mathcal{J}\subset [0,\infty)$.  Then, with $R^{-n}$ denoting the nth inverse iterate of $R$, the \emph{zeta function} of $R$ is defined for $Re(s)> d_R := \frac{2 \log N}{\log c}$ by
\[
 \zeta_{R,z_0} (s) = \lim_{n\rightarrow \infty}\sum_{z\in R^{-n}\{z_0\}}(c^n z)^{-\frac{s}{2}}.
\]  
\end{definition}

In addition, in the case of the Laplacian on the compact Sierpinski gasket, he discovered the product structure of the spectral zeta function that involves the zeta function of a polynomial.  Moreover, revisiting the example of fractal strings, he showed that the Riemann zeta function could be represented as the zeta function of a certain quadratic polynomial, thereby reinterpreting the corresponding product formula \eqref{La1} for the spectral zeta function of a fractal string obtained in \cite{Lapidus92}, \cite{Lapidus93}; see Remark \ref{TeplyaevRemark1} below.

\begin{theorem}[Teplyaev, \text{\cite{AT07}}]
\label{SG}
The spectral zeta function of the Laplacian on SG is 
\begin{equation}
\label{SpectralSG}
\zeta_{\Delta_{\mu}}(s) = \zeta_{R,\frac{3}{4}}(s)\frac{5^{-\frac{s}{2}}}{2} \bigg(\frac{1}{1-3 \cdot 5^{-\frac{s}{2}}}+\frac{3}{1-5^{-\frac{s}{2}}}\bigg)+\zeta_{R,\frac{5}{4}}\frac{5^{-s}}{2}\bigg(\frac{3}{1-3 \cdot 5^{-\frac{s}{2}}}-\frac{1}{1-5^{-\frac{s}{2}}}\bigg),
\end{equation}
where $R(z)=z(5-4z)$.  Furthermore, there exists $\epsilon >0$ such that $\zeta_{\Delta_{\mu}}(s)$ has a meromorphic continuation for $Re(s)>-\epsilon$, with poles contained in $\bigg\{\frac{2in\pi}{\log 5}, \frac{\log 9+2in\pi}{\log 5}: n\in \mathbb{Z}\bigg\}$.
\end{theorem}

\noindent We point out that some time later, and motivated in part by the results and conjectures of \cite{Kigami93}, \cite{Lapidus92}, \cite{Lapidus93}, but in part independently of \cite{AT04}, \cite{AT07}, Derfel, Grabner, and Vogel \cite{DGV08} have also worked on the same zeta function associated with a polynomial and proved that it has a meromorphic continuation on the whole complex plane.  They expressed the spectral zeta function in terms of this zeta function and a zeta function related to the generating set associated with the multiplicities of the eigenvalues of the operator. 

In closing this section, we mention the interesting example of the Cantor self-similar fractal string, which is the complement of the middle third Cantor set in $[0,1]$.  Hence, its geometric zeta function is given by $\zeta_{\mathcal{L}}(s)=(3^s-2)^{-1}$ and so the factorization formula for the associated spectral zeta function is $\zeta_{L}(s)=\pi^{-s}\zeta(s) (3^s-2)^{-1}$.  This product formula is similar to the one obtained for the spectral zeta function on $SG$.  More generally, one has analogous, although more complicated, expressions for $\zeta_{\mathcal{L}}$, the geometric zeta function of an arbitrary self-similar fractal string (see \cite{LvanF06}, Chs. 2 and 3), and hence, analogous factorization formulas for $\zeta_{L}$ discussed at the beginning of this introduction and obtained in \cite{Lapidus92}, \cite{Lapidus93}, \cite{LvanF06}; see Eq. \eqref{La1}.

\begin{remark}
\label{TeplyaevRemark1}
By considering the case of the unit interval, A. Teplyaev \emph{\cite{AT07}} proved that the Riemann zeta function can be described in terms of the zeta function of a quadratic polynomial of one complex variable.  More specifically, the Riemann zeta function $\zeta(s)$ can be represented by $\zeta(s)=\frac{1}{2}C^s \zeta_{R, 0}(s)$, where $C=\sqrt{2}\pi$ and $\zeta_{R, 0}(s)$ is the zeta function of the polynomial $R(z)=2z(2-z)$.
\end{remark}

\subsection{Hyperfunctions}
In this section, we give a brief introduction to hyperfunctions, which are the distributional generalization of analytic functions.  (We refer the interested reader to M. Sato's papers (\cite{MS58}, \cite{MS59}) for a general discussion of this beautiful subject, and to the books by U. Graf \cite{UG10} and M. Morimoto \cite{MM76} for a more elementary and directly accessible introduction.)

Let $\Omega$ be a subset of $\mathbb{R}$.  A complex neighborhood of $\Omega$ is an open subset $U \subset \mathbb{C}$ such that $\Omega$ is a closed subset of $U$.  We denote by $\mathcal{C}(U)$ and $\mathcal{C}(U\setminus\Omega)$ the vector spaces of holomorphic functions on $U$ and $U\setminus \Omega$, respectively.  The quotient space $\mathcal{C}(U\setminus \Omega)/ \mathcal{C}(U)$ is defined in terms of the equivalence relation according to which any holomorphic function in $\mathcal{C}(U\setminus \Omega)$ that extends holomorphically to all of $U$ is identified with the zero function.  By definition, each equivalence class represents a \textit{hyperfunction}. 

Aternatively, a hyperfunction on the real line, $f(x)=[F(z)]=[F_+, F_-]$, consists of two functions, $F_+(z)$ and $F_-(z)$, which are analytic in the upper and the lower half-planes, respectively,  and such that the following limit exists:
\[
\lim_{\epsilon \rightarrow 0^+} \bigg(F_+(x+i\epsilon)-F_-(x-i\epsilon)\bigg).
\] 
Every hyperfunction $[F_+, F_-]$ forms an equivalence class of the form $[h+ F_+, h+ F_-]$, where $h$ is a holomorphic function on $U$.  Any holomorphic function $g$ can be expressed as a hyperfunction $g= [g, 0]= [0, -g]$.  All the standard elementary operations on hyperfunctions are satisfied, such as 
\[
[F_+, F_-]+[G_+, G_-]= [F_+ +G_+, F_- +G_-] \mbox{  and  }
\frac{d}{dx}[F_+, F_-] =[\frac{d}{dx}F_+, \frac{d}{dx}F_-].
\]
The product of two hyperfunctions does not always make sense but the product of a hyperfunction $F$ by a holomorphic function $h$ can be defined by $h[F_+, F_-]=[h F_+, h F_-]$.

\begin{example}
\label{DracDelta}
We discuss here the\emph{ Dirac delta hyperfunction} on the real line $\mathbb{R}$;  it is one of the most important hyperfunctions and is given by $\delta_{\mathbb{R}}(z)=[-\frac{1}{2\pi iz}, -\frac{1}{2\pi iz}]$. For $x \neq 0$,

\begin{eqnarray*}
\begin{aligned}
f(x):=\lim_{\epsilon \rightarrow 0^+} \bigg(F_+(x+i\epsilon) -F_-(x-i\epsilon)\bigg)
&=\lim_{\epsilon \rightarrow 0^+}\bigg(\frac{-1}{2 \pi i (x+i\epsilon)} -\frac{-1}{2\pi i(x-i \epsilon)} \bigg)\\
 &= \lim_{\epsilon \rightarrow 0^+} \frac{\epsilon}{\pi(x^2+\epsilon^2)}
 = 0.
\end{aligned}
\end{eqnarray*}
\noindent For $x=0$, however, the above limit does not exist, and this is the point at which the delta `function' has an isolated singularity.  \emph{(}See \emph{\cite{UG10}}, \S 1.2  and \emph{\cite{MM76}}, \S 3.4.\emph{)}

\end{example}

Next, we will present a lemma which will be used to prove some of our main results.  We consider the bi-infinite series $\sum_{p=-\infty}^{\infty}(\gamma^p)^{-\frac{s}{2}}$, with $\gamma > 1$ and $s \in \mathbb{C}$.  The most peculiar behavior of this series is that it seems to be equal to zero, even though it does not make sense to sum up a complex series with one part converging and the other part diverging.  However, this series can be represented by a suitable version of the Dirac hyperfunction, as we shall soon see.  For now, we carry out the naive computation as follows:

\[
\sum_{p=-\infty}^{\infty}(\gamma^p)^{-\frac{s}{2}}
= \sum_{p=-\infty}^{-1}(\gamma^p)^{-\frac{s}{2}}+ \sum_{p=0}^{\infty}(\gamma^p)^{-\frac{s}{2}}
=\frac{\gamma^{\frac{s}{2}}}{1-\gamma^{\frac{s}{2}}} + \frac{1}{1-\gamma^{-\frac{s}{2}}} = \frac{1}{1-\gamma^{-\frac{s}{2}}} - \frac{1}{1-\gamma^{-\frac{s}{2}}}=0.
\]

Note that this computation is meaningless, unless it is properly interpreted.  Indeed, we have added two infinite series, one of which is convergent only for $Re(s) >0$, whereas the other series is convergent only for $Re(s)<0$.  In fact, fortunately, the geometric part $\sum_{p=-\infty}^{\infty} (\gamma^p)^{-\frac{s}{2}}$ can be interpreted in terms of the \emph{Dirac delta hyperfunction on the unit circle}, $\delta_{\mathbb{T}}(w)=[\delta^+_{\mathbb{T}}(w), \delta^-_{\mathbb{T}}(w)]$, by means of a suitable change of variable; namely, $w=\gamma^{-\frac{s}{2}}$.  Recall that the Dirac hyperfunction on the unit circle $\mathbb{T}$ is defined as $\delta_{\mathbb{T}}=[\delta^+_{\mathbb{T}}, \delta^-_{\mathbb{T}}]=[\frac{1}{1-z}, \frac{1}{z-1}]$.  It consists of two analytic functions, $\delta^+_{\mathbb{T}}: E \rightarrow \mathbb{C}$ and $\delta^-_{\mathbb{T}}: \mathbb{C} \backslash \bar{E} \rightarrow \mathbb{C}$, where $E=\{z\in \mathbb{C}: |z| <1 + \frac{1}{N}\}$ for a large natural number $N$.  In other words, a hyperfunction on $\mathbb{T}$ can be viewed as a suitable pair of holomorphic functions, one defined on the unit disk $|z| <1$, and one on its exterior, $|z|>1$.
(See, for example, \cite{UG10}, \S 1.3 and \cite{MM76}, \S3.3.2 for a discussion of various changes of variables in a hyperfunction.  Moreover, see \cite{YT87} for a detailed discussion of $\delta_{\mathbb{T}}$ and, more generally, of hyperfunctions on $\mathbb{T}$.)

\begin{lemma}
\label{DeltaHyperfunction}
Let $\gamma >1$ and $\delta_{\mathbb{T}}$ be the Dirac hyperfunction on $\mathbb{T}$, as above.  Then
\begin{equation}
\label{Hperfunction}
\sum_{p=-\infty}^{\infty} (\gamma^p)^{-\frac{s}{2}}=\delta_{\mathbb{T}}(\gamma^{-\frac{s}{2}}).
\end{equation}
\end{lemma}

\begin{proof}

We introduce the formal expression  $\psi(w)=\sum_{p=-\infty}^{\infty}w^p$.  Note that
$$
\psi(w)=
\begin{cases}
\sum_{p=0}^{\infty}w^p =\frac{1}{1-w}, & \text{if }|w|<1\\
\sum_{p=-\infty}^{-1}w^p=\frac{1}{w-1}, & \text{if }|w| >1. 
\end{cases}
$$
Now, essentially by definition (see \cite{YT87}), $\psi(w)= \delta_{\mathbb{T}}(w)= [\delta^{+}_{\mathbb{T}}(w), \delta^{-}_{\mathbb{T}}(w)]$.  We conclude the proof of \eqref{Hperfunction} by making the change of variable $w=\gamma^{-\frac{s}{2}}$ and noting that $|w| <1$ and $|w| >1$ correspond to the upper and lower half-planes $Re(s)>0$ and $Re(s) <0$, respectively.  Indeed, $\log \gamma >0$ since $\gamma > 1$.
 \end{proof}

\subsection{Results for the Infinite Sierpinski Gasket}
We now extend the finite Sierpinski gasket to the infinite (or unbounded) Sierpinski gasket.  Let $k=\{k_n\}_{n\geq1}$ be a fixed sequence, with $k_n \in \{1, 2, 3\}$ for all $n\geq1$.  We construct a sequence $SG^{(n)}= \Psi_{k, n}^{-1}(SG)$, where $\Psi_{k, n}=\Psi_{k_1...k_n}:=\Psi_{k_n}\circ...\circ \Psi_{k_1}$.  The infinite Sierpinski gasket is then defined by 
\[
SG^{(\infty)}=\bigcup_{n=0}^{\infty}SG^{(n)},
\]
viewed as a blow-up of SG.  The mth pre-gasket approximating $SG^{(n)}$ and $SG^{(\infty)}$ are
$V^{(n)}_m=\Psi_{k, n}^{-1}(V_{n+m})$ and $V^{(\infty)}_m= \cup_{n=0}^{\infty}  \Psi^{-1}_{k, n}(V_{n+m})$, respectively.  Note that $V_m=V_m^{(0)}$ and $SG=SG^{(0)}$.  We next define the Laplacian $\Delta^{(n)}$ on $SG^{(n)}$ as follows: $\Delta^{(n)} u=f \in L^2(SG^{(n)}, \mu)$ iff $\mathcal{E}_{SG^{(n)}}(u, v)=\int_{SG^{(n)}} \Delta^{(n)} u v d\mu$, where $\mathcal{E}_{SG^{(n)}}$ is a scaled copy of $\mathcal{E}_{n+m}$ on $V_{n+m}$ for the finite Sierpinski pre-gasket.  The pointwise Laplacian $\Delta^{(\infty)}$ on $SG^{(\infty)}$ can then be defined by the following pointwise limit: $5^n\Delta^{(n)} u \rightarrow \Delta^{(\infty)}u$ as $n \rightarrow \infty$ (see \cite{FS92},  \cite{Kigami01}).  Here, it should be noted that the Laplacian on the finite Sierpinski gasket coincides with the restriction of $\Delta^{(\infty)}$ to the interior of the finite (or bounded) SG.  It follows from this construction that the spectrum of the Laplacian $\Delta^{(\infty)}$ on the infinite Sierpinski gasket can be generated by the spectrum of the Laplacian on the finite SG:

\begin{theorem}[see, e.g., \text{\cite{AT98}}]
Let $R(z)=z(5-4z)$, as in Thm. \ref{SG}.  Then the spectrum of the self-adjoint operator $\Delta^{(\infty)}$ acting on $L^2(SG^{(\infty)}, \mu)$ is pure point and the set of compactly supported eigenfunctions is complete.  Furthermore, the set of eigenvalues is given by $\bigcup_{n=-\infty}^{\infty} 5^n \mathcal{R}\{\Sigma\}$,
where $\Sigma=\{\frac{3}{2}\} \cup (\cup_{j=0}^{\infty} R^{-j}\{\frac{3}{4}\}) \cup (\cup_{j=0}^{\infty}R^{-j}\{\frac{5}{4}\})$ is the set of eigenvalues of the Laplacian $\Delta_{\mu}$ on the finite SG, $\mathcal{R}(z):=\lim_{m \rightarrow \infty} 5^m R^{-m}_{-}(z)$ and $R^{-m}_{-}$ is the branch of the mth inverse iterate of $R$ that passes through the origin.
\end{theorem}

In particular, the spectrum of $\Delta^{(\infty)}$ has the following form: 
\begin{equation}
\label{InfiniteSGSpectral}
\bigcup_{n=-\infty}^{\infty} \bigcup_{j=0}^{\infty} 5^n \mathcal{R}(R^{-j}(z_0)),
\end{equation}
where $z_0=\frac{3}{4}, \frac{5}{4}$.  Every eigenvalue $\lambda$ of $\Delta^{(\infty)}$ can be expressed as $\lambda=5^n\lim_{m\rightarrow \infty}5^mR^{-m}_{-}(z_m)$ for some $n \in \mathbb{Z}$, with $z_m$ in the spectrum $\sigma(\Delta_m)$ of the finite mth Sierpinski pre-gasket.

We now state our first main result:
\begin{theorem}
\label{infiniteSGzeta}
The spectral zeta function $ \zeta_{\Delta^{(\infty)}}$ of the Laplacian $\Delta^{(\infty)}$ on the infinite Sierpinski gasket $SG^{(\infty)}$ is given by 
\begin{equation}
\label{infiniteSG}
\zeta_{\Delta^{(\infty)}}(s)=\delta_{\mathbb{T}}(5^{-\frac{s}{2}}) \zeta_{\Delta_{\mu}}(s),
\end{equation}
where $\delta_{\mathbb{T}}$ is the Dirac hyperfunction \emph{(}as in \S 2.1\emph{)} and $\zeta_{\Delta_{\mu}}$ is the spectral zeta function of the Laplacian on the finite SG as given and factorized explicitly in Eq. \eqref{SpectralSG} of Thm. \ref{SG}.
\end{theorem}

\begin{proof}
In light of Eq. \eqref{InfiniteSGSpectral}  and Thm. \ref{SG}, this result follows from Lemma \ref{DeltaHyperfunction}.  Indeed, note that $\lambda$ is an eigenvalue of $\Delta_{\mu}$ iff $5^n \lambda$ is an eigenvalue of $\Delta^{(n)}$.  Furthermore, if $\lambda$ is an eigenvalue of $\Delta^{(\infty)}$, then $5^n \lambda$ is also an eigenvalue of $\Delta^{(\infty)}$ for $n \in \mathbb{Z}$ (see \cite{FS92}).  Combining these two facts, we obtain $\zeta(s)=\bigg(\sum_{n=-\infty}^{\infty} (5^{-\frac{s}{2}})^n \bigg) \zeta_{\Delta_{\mu}}(s)$.  By Lemma \ref{DeltaHyperfunction} and using the substitution $\gamma=5$, we now deduce the desired result.
\end{proof}

\section{The Sturm--Liouville Operator}

\subsection{Dirichlet forms and the Sturm--Liouville operator on [0,1]}
We investigate a class of self-similar sets and measures in terms of the spectrum and the spectral zeta function of the associated fractal differential operators.  C. Sabot, in a series of papers (\cite{CS98}--\cite{CS05}), extended the decimation method to  Laplacians defined on a class of finitely-ramified self-similar sets with blow-ups.  We discuss the prototypical example he studied, fractal Laplacians on the blow-up $I_{<\infty>}=[0, \infty)$ of the unit interval $I=I_{<0>}=[0,1]$.  From now on, we will assume that 
\begin{equation}
\label{constants}
0<\alpha<1 \mbox{,  } b=1-\alpha, \mbox{  } \delta =\frac{\alpha}{1-\alpha}, \mbox{ and } \gamma=\frac{1}{\alpha(1-\alpha)}. 
\end{equation}
Consider the contraction mappings from $I=[0, 1]$ to itself given by
\[
\Psi_1(x)= \alpha x, \mbox{  } \Psi_2(x)=1-(1-\alpha)(1-x),
\]
and the unique self-similar measure $m$ on $[0, 1]$ such that for all $f \in C([0, 1])$, 
\begin{equation}
\label{IntegralIdentity}
\int_0^1 f dm = b\int_0^1 f\circ\Psi_1 dm + (1-b)\int_0^1 f\circ\Psi_2 dm.
\end{equation}

Define $H_{<0>} = -\frac{d}{dm}\frac{d}{dx}$, the free Hamiltonian with Dirichlet boundary conditions on $[0, 1]$, by $H_{<0>} f= g$ on the domain 
\[
\bigg\{ f\in L^2(I, m), \exists g \in L^2(I, m), f(x)= cx + d + \int_0^x \int_0^y g(z)dm(z) dy, f(0)= f(1)=0 \bigg\}.  
\]
The operator $H_{<0>}$  is the infinitesimal generator associated with the Dirichlet form $(a,\mathcal{D})$ given by $a(f,g) = \int_0^1 f^{\prime} g^{\prime}dx$, for $f,g \in \mathcal{D}$,  where $\mathcal{D} = \{ f \in L^2(I, m): f^{\prime} \in L^2(I, dx)\}$.  As can be easily checked, $a$ satisfies the self-similarity equation
\begin{equation}
\label{SelfsimilarIdentity}
a(f)=\alpha^{-1} a(f \circ \Psi_1) + (1-\alpha)^{-1} a(f \circ \Psi_2),
\end{equation}
where we denote the quadratic form $a(f, f)$ by $a(f)$.
(See, e.g., \cite{UF03} for an exposition.)

Next, extend $I$ to $I_{<n>} = \Psi_1^{-n}(I)=[0,\alpha^{-n}]$, which can be expressed as a self-similar set as follows:
$I_{<n>} = \bigcup_{i_1,...,i_n}\Psi_{i_1...i_n}(I_{<n>})$, where $(i_1,...,i_n) \in {\{1,2\}}^{n}$.  Here, we have set $\Psi_{i_1...i_n}=\Psi_{i_n} \circ ...\circ \Psi_{i_1}$.  We define the self-similar measure $m_{<n>}$ by $\int_{I_{<n>}}fdm_{<n>} = (1-\alpha)^{-n} \int_I f\circ\Psi_1^{-n}dm$, for all $f \in C(I_{<n>})$.  Similarly, the corresponding differential operator, $H_{<n>}= -\frac{d}{dm_{<n>}}\frac{d}{dx}$ on $I_{<n>}=[0, \alpha^{-n}]$, can be defined  as the infinitesimal generator of the Dirichlet form
$(a_{<n>},\mathcal{D}_{<n>})$ given by $a_{<n>}(f) = \int_0^{{\alpha}^{-n}} (f')^2 dx = \alpha^n a(f\circ\Psi_1^{-n}), \mbox{ for }  f \in \mathcal{D}_{<n>}$,
where $\mathcal{D}_{<n>} = \{f \in L^2(I_{<n>}, m_{<n>}):  f' \mbox{ exists and } f' \in L^2(I_{<n>}, dx)\}$.

We define $H_{<\infty>}$ as the operator $ -\frac{d}{dm_{<\infty>}}\frac{d}{dx}$ with Dirichlet boundary conditions on $I_{<\infty>}= [0, \infty)$.  It is clear that the (projective system of) measures $m_{<n>}$ give rise to a measure $m_{<\infty>}$ on $I_{<\infty>}$ since for any $f \in \mathcal{D}_{<n>}$  with $supp(f) \subset [0, 1]$, $a_{<n>}(f)=a(f)$ and $\int_{I_{<n>}} fdm_{<n>} =\int_{I} f dm$.  Furthermore, we define the corresponding Dirichlet form $(a_{<\infty>}, \mathcal{D}_{<\infty>})$ by $a_{<\infty>}(f)=\lim_{n \rightarrow \infty} a_{<n>}(f|_{I_{<n>}}), \mbox{ for } f \in \mathcal{D}_{<\infty>}$,
where $\mathcal{D}_{<\infty>}=  \{f \in L^2(I_{<\infty>}, m_{<\infty>}): \sup_n  a_{<n>}(f|_{I_{<n>}}) < \infty\}$.  Clearly, $a_{<\infty>}$ satisfies a self-similar identity analogous to Eq. \eqref{SelfsimilarIdentity} and its infinitesimal generator is $H_{<\infty>}$.

The study of the eigenvalue problem 
\begin{equation}
\label{EigenProb}
H_{<n>}f = -\frac{d}{dm_{<n>}}\frac{d}{dx}f = \lambda f
\end{equation}
for the Sturm--Liouville operator with Dirichlet boundary conditions on $I_{<n>}$ revolves around a map $\rho$, called the \textit{renormalization map}, which is initially defined on a space of quadratic forms associated with the fractal and then, via analytic continuation, on $\mathbb{C}^3$ as well as (by homogeneity) on $\mathbb{P}^2({\mathbb{C}})$.  The propagator of the above differential equation \eqref{EigenProb} is very useful in producing this rational map,
\begin{equation}
\label{RenormalizationMap}
\rho([x,y,z])= [x(x+\delta^{-1}y)-\delta^{-1}z^2, \delta y(x+\delta^{-1}y)-\delta z^2, z^2],
\end{equation}
defined on the complex projective plane $\mathbb{P}^2(\mathbb{C})$.  Here, $[x, y, z]$ denote the homogeneous coordinates of a point in $\mathbb{P}^2(\mathbb{C})$, where $(x, y, z) \in \mathbb{C}^3$ is identified with $(\beta x, \beta y, \beta z)$ for any $\beta \in \mathbb{C}$, $\beta \neq 0$.  Note that in the present case, $\rho$ is a homogeneous polynomial of total degree two.  As we shall see later on, the spectrum of the fractal Sturm--Liouville operator is intimately related to the iteration of $\rho$.  In the sequel, we shall assume that $\delta \leq1$ in order for the spectrum of $H_{<0>}$, $H_{<n>}$ ($n=1, 2,...$) and $H_{<\infty>}$ to be purely discrete.

We define the \textit{propagator} $\Gamma_{\lambda}(s, t)$ for the eigenvalue problem $-\frac{d}{dm_{<\infty>}}\frac{d}{dx}f=\lambda f$ associated with the operator $H_{<\infty>}$ on $I_{<\infty>}=[0, \infty)$ as a time evolution function which for each $0\leq s \leq t$ is a $2 \times 2$ matrix with nonzero determinant such that the solution of the equation satisfies 

\[
\left[ \begin{array}{c} f(t)\\ f^{\prime}(t) \end{array} \right] = \Gamma_{\lambda}(s, t)  \left[ \begin{array}{c} f(s) \\ f^{\prime}(s) \end{array} \right].
\]

\noindent Using the self-similarity relations \eqref{IntegralIdentity} and \eqref{SelfsimilarIdentity} satisfied by the measure $m$ and the Dirichlet form $a$, respectively, and recalling that $\gamma$ is given by Eq. \eqref{constants}, we obtain $\Gamma_{<n>, \lambda}=D_{\alpha^n} \circ \Gamma_{\gamma^n\lambda} \circ D_{\alpha^{-n}}$ for the eigenvalue problem  $-\frac{d}{dm_{<n>}}\frac{d}{dx}f=\lambda f$, where
\[
D_{\alpha^n}=\begin{bmatrix} 1 & 0 \\ 0 & \alpha^n \end{bmatrix}. 
\]

The dynamics of the renormalization map $\rho$ plays a key role in calculating the spectrum of the operators $H_{<n>}$.  We introduce the \textit{invariant curve} $\phi$, which is holomorphic on $\mathbb{C}$ and satisfies the following key functional equation:
\begin{equation}
\label{InvariantCurve}
\rho(\phi(\lambda)) = \phi(\gamma\lambda),
\end{equation}
for all $\lambda \in \mathbb{C}$. An \textit{attractive fixed point} $x_0$ of $\rho$ is a point such that $\rho x_0 = x_0$ and for any other point $x$ in some neighborhood of $x_0$, the sequence $\{\rho^nx\}_{n=0}^{\infty}$ converges to $x_0$.  The \textit{basin of attraction} of a fixed point is contained in the Fatou set of $\rho$.  For $\delta >1$, $x_0= [0, 1, 0]$ is an attractive fixed point of $\rho$.  The set 
\begin{equation}
\label{SetD}
D = \{[x,y,z]: x+\delta^{-1} y = 0\}
\end{equation}
is part of the Fatou set of $\rho$ since it is contained in the basin of attraction of $x_0$.  (For various notions of higher-dimensional complex dynamics, we point out, e.g., \cite{JF96} and \cite{FS94}.)  The set $D$ and the invariant curve $\phi$ together determine the spectrum of $H_{<n>}$ and of $H_{<\infty>}$.  Moreover, the set of eigenvalues can be described by the set 

\begin{equation}
\label{SetS}
S = \{\lambda \in \mathbb{C}: \phi(\gamma^{-1}\lambda) \in D\},
\end{equation}
the `time intersections' of the curve $\phi(\gamma^{-1} \lambda)$ with $D$.  It turns out that $S$ is countably infinite and contained in $\mathbb{R}^+$.    We write $S= \{\lambda_k\}^{\infty}_{k=1}$, with $\lambda_1 \leq \lambda_2 \leq... \leq \lambda_k \leq...$ repeated accordingly to multiplicity.  Furthermore, we call $S$ the \textit{generating set} for the spectrum of $H_{<n>}$, with $n=0,1,...,\infty$. 

Let $S_p = \gamma^p S$, for each $p \in \mathbb{Z}$.  As was noted earlier, the spectrum of $H_{<\infty>}$ with Dirichlet boundary conditions is pure point for $\alpha \leq \frac{1}{2}$ (hence, for $\delta \leq1$ and $\gamma \geq 4$), an hypothesis we will make from now on, and can be deduced from the spectrum of $H_{<0>}$:

\begin{theorem}[Sabot, \text{\cite{CS05}}]
\label{Spectrum}
The spectrum of $H_{<0>}$ on $I =I_{<0>}$ is $\bigcup_{p=0}^{\infty} S_p$ and the spectrum of $H_{<\infty>}$ on $\mathbb{R}^+$ is $\bigcup_{p=-\infty}^{\infty}S_p$.  Moreover, for any $n \geq 0$, the spectrum of $H_{<n>}$ is equal to $\bigcup_{p=-n}^{\infty}S_p$.  For $n=0, 1, ...\infty$, each eigenvalues of $H_{<n>}$ is simple.
\end{theorem}

\noindent The diagram of the set of eigenvalues of the operator $H_{<\infty>}$ is as follows:

\[\begin{array}{ccccc}
\vdots & \vdots & \vdots & \vdots & \\
\gamma^{-2}\lambda_1 & \gamma^{-2}\lambda_2 & \gamma^{-2}\lambda_3 & \gamma^{-2}\lambda_4 & \cdots \\
\gamma^{-1}\lambda_1 & \gamma^{-1}\lambda_2 & \gamma^{-1}\lambda_3 & \gamma^{-1}\lambda_4 & \cdots\\

\lambda_1 & \lambda_2 & \lambda_3 & \lambda_4 & \cdots \\
\gamma\lambda_1 & \gamma\lambda_2 & \gamma\lambda_3 & \gamma\lambda_4 & \cdots \\
\gamma^2\lambda_1 & \gamma^2\lambda_2 & \gamma^2\lambda_3 & \gamma^2\lambda_4 & \cdots\\
\vdots & \vdots & \vdots & \vdots & \\
\end{array}
\]

Sabot's work (\cite{CS98}--\cite{CS03}) has sparked an interest in generalizing the decimation method to a broader class of fractals and therefore, to the iteration of rational functions of several complex variables.  

\begin{theorem}[Sabot, \text{\cite{CS05}}]
\label{EigenValuesFunctions}
Given any $k \geq 1$, if $f_{k}$ is the normalized solution of the equation $H_{<\infty>}f = \lambda_{k} f$ for $\lambda_{k} \in S$ \emph{(}i.e., if $f$ is an eigenfunction of $H_{<\infty>}$ with eigenvalue $\lambda_k$\emph{)}, then $f_{k,p}:=f_{k}\circ\Psi_1^{-p}$ is the solution of the equation $H_{<\infty>}f = \lambda_{k,p}f$, where $\lambda_{k,p} := \gamma^p\lambda_k$ and $p \in \mathbb{Z}$ is arbitrary.
Moreover, for each $n \geq 1$, $f_{k,p,<n>}:=f_{k,p}|_{I_{<n>}}$, the restriction of $f_{k,p}$ to $I_{<n>}$, is then the solution of the equation $H_{<n>}f = \lambda_{k,p}f$.  

Finally, for each fixed $n \geq 0$,
$
\{f_{k, p, <n>}: k\geq 1, p \geq -n\}
$
is a complete set of eigenfunctions of $H_{<n>}$ in the complex Hilbert space $L^2(\mathbb{R}^+, m_{<\infty>})$.
\end{theorem}

There are a number of fractals for which the decimation method has been established or explored.  The interested readers can consult the following references by Shima \cite{TS96}, Fukushima--Shima \cite{FS92}, Kigami--Lapidus \cite{KL01}, Strichartz \cite{RS06}, Bajorin \textit{et al.} (\cite{BCD08}, \cite{BC08}), Teplyaev (\cite{AT04}, \cite{AT07}) and Derfel \textit{et al.} \cite{DGV08} in the case of rational functions of a single complex variable, and by Sabot (\cite{CS98}--\cite{CS03}) in the significantly more general case of rational functions of several complex variables.

\vspace{3mm}

\subsection{The zeta function associated with the renormalization map}

We now introduce a multivariable analog of the polynomial zeta function of Def. \ref{PolyZeta}.
\begin{definition}
\label{MultiPolyZeta}
For Re(s) sufficiently large, we define the \emph{zeta function of the renormalization map} $\rho$ to be
\begin{equation}
\label{MultiPoly}
\zeta_{\rho}(s) = \sum_{p=0}^{\infty}\sum_{\{\lambda \in \mathbb{C}:\hspace{.5mm } \rho^p(\phi(\gamma^{-(p+1)}\lambda)) \in D\}} (\gamma^p \lambda)^{-\frac{s}{2}}.
\end{equation}
\end{definition}

We can now state our first result regarding the fractal Sturm--Liouville operator:

\begin{theorem}
\label{H0MultiPolyZeta}
The zeta function $\zeta_{\rho}(s)$ of the renormalization map $\rho$ is equal to the spectral zeta function of $H_{<0>}$ \emph{(}as given by the $n=0$ case of Prop. \ref{HnZetaFunction} below\emph{):} $\zeta_{\rho}(s)=\zeta_{H_{<0>}}(s)$. 
\end{theorem}

\begin{proof}
We have successively:
\[
\zeta_{\rho}(s) = \sum_{p=0}^{\infty}\sum_{\{\lambda \in \mathbb{C}: \hspace{.5mm }\rho^p(\phi(\gamma^{-(p+1)}\lambda)) \in D\}} (\gamma^p \lambda)^{-\frac{s}{2}} = \sum^{\infty}_{p=0} \sum^{\infty}_{\lambda \in S} (\gamma^p\lambda)^{-\frac{s}{2}} = \zeta_{H_{<0>}}(s).
\]
\noindent In order to justify the first equality, we show that the set $\{\lambda \in \mathbb{C}: \rho^p(\phi(\gamma^{-(p+1)}\lambda)) \in D\}$ is exactly equal to $S= \{\lambda \in \mathbb{C}: \phi(\gamma^{-1}\lambda) \in D\}$.  Recall from Eq.\eqref{InvariantCurve} the relation $\rho(\phi(\lambda)) = \phi(\gamma\lambda)$, for all $\lambda \in \mathbb{C}$.  After $p$ iterations, this equation becomes $\rho^p(\phi(\lambda)) = \phi(\gamma^p\lambda)$.  Therefore, we get $\rho^p(\phi(\gamma^{-(p+1)}\lambda)) = \phi(\gamma^p\gamma^{-p-1}\lambda) = \phi(\gamma^{-1}\lambda)$, for $p=0, 1, 2...$ 
\end{proof}

We have a sequence of operators $H_{<n>}=-\frac{d}{dm_{<n>}}\frac{d}{dx},$ starting with $H_{<0>}$  on $[0, 1]$, which converges to the Sturm--Liouville operator $H_{<\infty>}$ on $[0, \infty)$.  We will now consider the associated spectral zeta functions and their product formulas.  Recall that given an integer $n\geq 0$, the spectral zeta $\zeta_{H_{<n>}}(s)$ of $H_{<n>}$ on $[0, \alpha^{-n}]$ is $\zeta_{H_{<n>}}(s) = \sum_{\lambda \in S} \sum_{p=-n}^{\infty} (\gamma^p\lambda)^{-\frac{s}{2}}$.  Then, a simple computation yields the following result:

\begin{proposition}
\label{HnZetaFunction}
For $n \geq 0$ and $Re(s)$ sufficiently large, we have
\begin{equation}
\label{HnZeta}
\zeta_{H_{<n>}}(s)=\frac{(\gamma^n)^{\frac{s}{2}}}{1-\gamma^{-\frac{s}{2}}} \zeta_S(s),
\end{equation}
where $\zeta_S(s)$ is the \emph{geometric zeta function of the generating set} $S$.  Namely, $\zeta_S(s):= \sum_{j=1}^{\infty}(\lambda_j)^{-\frac{s}{2}}$ \emph{(}for Re\emph{(}s\emph{)} large enough\emph{)} or is given by its meromorphic continuation thereof.\end{proposition}


In the case of the operator $H_{<\infty>}$, the geometric part of the product formula for the spectral zeta function $\zeta_{H_{<\infty>}}$ is given by the Dirac hyperfunction $\delta_{\mathbb{T}}$ (see \S 2.1).

\begin{theorem}
\label{HZetaFunction}
The spectral zeta function $\zeta_{H_{<\infty>}}$ is factorized as follows\emph{:}
\begin{equation}
\label{HZeta}
\zeta_{H_{<\infty>}}(s)= \delta_{\mathbb{T}}(\gamma^{-\frac{s}{2}}) \cdot \zeta_S(s).
\end{equation}
\end{theorem}

\begin{proof}
In light of Lemma \ref{DeltaHyperfunction} (and since $\gamma >1$), we have
\[
\zeta_{H_{<\infty>}}(s) = \sum_{p=-\infty}^{\infty}(\gamma^p)^{-\frac{s}{2}} \sum_{j=1}^{\infty} (\lambda_j)^{-\frac{s}{2}} =\delta_{\mathbb{T}}(\gamma^{-\frac{s}{2}}) \cdot \zeta_S(s).\qedhere
\]
\end{proof}

Next,  we revisit and extend some of the earlier results obtained in \cite{AT07}.  More precisely, we show that the zeta function associated with the renormalization map coincides with the Riemann zeta function for a special value of $\alpha$.

\subsubsection{The case $\alpha= \frac{1}{2}$: Connection with the Riemann zeta function}

When $\alpha= \frac{1}{2}$, the self-similar measure $m$ coincides with Lebesgue measure on $[0,1]$ and hence, $H=H_{<0>}$ coincides with the usual Dirichlet Laplacian on the unit interval $[0, 1]$. 

\begin{theorem}
\label{RiemannZetaFunction}
When $\alpha=\frac{1}{2}$, the Riemann zeta function $\zeta$ is equal \emph{(}up to a trivial factor\emph{)} to the zeta function $\zeta_{\rho}$ associated with the renormalization map $\rho$ on $\mathbb{P}^2(\mathbb{C})$.  More specifically, 
\begin{equation}
\label{RiemannZeta}
\zeta(s) =\pi^s \zeta_{\rho}(s)= \frac{\pi^s}{1-2^{-s}} \mbox{ }\zeta_S(s),
\end{equation}
where $\zeta_{\rho}$ is given by Def. \ref{MultiPolyZeta} and the polynomial map $\rho: \mathbb{P}^2(\mathbb{C}) \rightarrow \mathbb{P}^2(\mathbb{C})$ is given by Eq. \eqref{RenormalizationMap} with $\alpha=\frac{1}{2}$ \emph{(}and hence, in light of Eq. \eqref{constants}, with $\delta=1$ and $\gamma=4$\emph{):}
\begin{equation}
\label{RenormalMap}
\rho([x,y,z])= [x(x+y)-z^2, y(x+y)-z^2, z^2].
\end{equation}

\end{theorem}

\begin{proof}
First, by Thm. \ref{H0MultiPolyZeta} and Prop. \ref{HnZetaFunction} (with n=0), we have (since $\gamma=4$) $\zeta_{H_{<0>}}(s) =\zeta_{\rho}(s)=(1-2^{-s})^{-1} \mbox{ } \zeta_S(s)$.

Next, we recall that the eigenvalues of the Dirichlet Laplacian $L=-\frac{d^2}{dx^2}$  on $[0, 1]$ are $\kappa_j= \pi^2 j^2$, for $j=1,2,...$.  Thus, in light of Def. \ref{SpecZetaFunction}, the associated spectral zeta function is $\zeta_L(s)=\sum_{j=1}^{\infty}(\pi^2j^2)^{-\frac{s}{2}} = \pi^{-s} \zeta(s)$.  Note that in the present case, the Sturm--Liouville operator $H_{<0>}$ and the Dirichlet Laplacian $L$ on $[0, 1]$ coincide; hence, the corresponding spectral zeta functions are equal: $\zeta_{H_{<0>}}(s) = \zeta_L(s)$.  In light of Thm. \ref{H0MultiPolyZeta}, $\zeta_{H_{<0>}}(s)=\zeta_{\rho}(s)$ and we therefore obtain the relation $\zeta(s) = \pi^s \zeta_{\rho}(s)$, with $\zeta _{\rho}$ given by Eq. \eqref{MultiPoly} and $\rho$ defined by Eq. \eqref{RenormalMap}, as desired.
\end{proof}

\begin{remark}
\label{TeplyaevRemark}
This is an extension to several complex variables of A. Teplyaev's result \emph{\cite{AT07}} discussed in Remark \ref{TeplyaevRemark1} above.
\end{remark}

\begin{remark}
\label{LapidusRemark}
Still assuming that $\alpha=\frac{1}{2}$ and since Eq. \eqref{RiemannZeta} implies that $\zeta_{\rho}(s)=\pi^{-s} \zeta(s)$, we deduce that the factorization formula \eqref{La1} for the spectral zeta function $\zeta_{L}(s)=\zeta_{sp}(s)$ of a fractal string $\mathcal{L}$ can be rewritten as follows\emph{:}
\begin{equation}
\label{La3}
\zeta_{L}(s)=\zeta_{\rho}(s) \cdot \zeta_{\mathcal{L}}(s),
\end{equation}
in agreement with Eq. \eqref{La2}.  \emph{(}Compare with \emph{\cite{Lapidus92}}, \emph{\cite{Lapidus93}} and \emph{\cite{LvanF06}}, Thm. 1.19.\emph{)}  Here, $\rho$ is the homogeneous quadratic polynomial on $\mathbb{C}^3$ \emph{(}or rather, on $\mathbb{P}^2(\mathbb{C})$\emph{)} given by Eq. \eqref{RenormalMap}.

\end{remark}

\section{Concluding Remarks}
We expect to obtain a similar product structure for other fractal Laplacians on Sabot decimable self-similar fractals.  Our immediate aim would be to apply the obtained results to the modified Koch curve for which the decimation method is well established with a rational map of one complex variable.  Another goal would be to analyze a large class of finitely-ramified self-similar sets with (possibly random) blow-ups; a special case of that is the infinite Sierpinski gasket, which corresponds to the blow-up of the classic Sierpinski gasket and was studied in the deterministic case in \S 2.2.  A study of random infinite gaskets and other self-similar fractals, as in \cite{CS03} but along the lines of \S 2.2 and \S 3.2, remains to be carried out.

Using Sabot's multivariable extension of the decimation method, one should be able to obtain an analogous factorization formula for such fractals.  Such a generalization would also enable us to better understand the nature of the spectrum of the Laplacian and to formulate and possibly solve suitable direct and inverse spectral problems in this context.  In a more familiar language, and appropriately interpreted, this would enable us in certain situations to ``hear the shape of a fractal drum".  (See, e.g., \cite{Kigam93}--\cite{LvanF06}.)
 
Furthermore, it would be interesting, both mathematically and physically, to obtain related results in the situation where the Laplacian under investigation has a continuous spectrum or, more generally, a combination of continuous and discrete spectra.  We would then have to work with a suitably defined notion of density of states, both at the geometric and spectral levels.  (Compare, e.g., \cite{LvanF06}, \S6.3.1 and \cite{Kigam93}, \cite{CS01}.)

Moreover, as we have seen, Sabot discovered in \cite{CS98}--\cite{CS03} some striking relationships between the spectral properties of certain differential operators on fractals and the iteration of rational maps of several complex variables.  The further study of the connections between these rational maps and the spectral zeta functions of fractal Laplacians is one of the main goals of future research on this topic and should lead to a deeper exploration of complex dynamics in higher dimensions, in relation to the spectral theory of fractal drums.  It may also have applications to condensed matter physics (\cite{AO82}, \cite{B79}, \cite{HHW87}, \cite{NY03}--\cite{RT83}), particularly, the study of random and fractal media.

\vspace{7mm}
\noindent \thanks{\textbf {Acknowledgements}.}
We wish to thank Michael Maroun for a helpful conversation about hyperfunctions.  The research of Michel Lapidus was  partially supported by the National Science Foundation under the research grants DMS-0707524 and DMS-1107750.  Michel Lapidus would like to thank the Institut des Hautes Etudes Scientifiques (IHES) where he was a visiting professor while this paper was completed.

\nocite{JS06}

\end{document}